\documentclass[10pt, draftclsnofoot, onecolumn]{IEEEtran}
\usepackage{amsthm}
\usepackage{amsmath}
\usepackage{amsfonts}
\usepackage{graphicx}
\usepackage{latexsym}
\usepackage{amssymb}
\usepackage{stmaryrd} 
\usepackage{comment}
\usepackage{amscd}
\usepackage{bm}
\usepackage{hyperref}
\newcommand\myshade{70} 
\hypersetup{ 
	linkcolor  = red!\myshade!black,
	citecolor  = blue!\myshade!black,
	urlcolor   = blue!\myshade!black,
	colorlinks = true,
}
\usepackage[small]{caption}
\usepackage[dvipsnames]{xcolor}
\usepackage[noend]{algpseudocode}
\algrenewcommand\alglinenumber[1]{\scriptsize #1:}
\algrenewcommand\algorithmicindent{1em}%
\allowdisplaybreaks
\usepackage{graphicx}
\usepackage{subfigure}
\usepackage{wrapfig}
\usepackage{float}

\usepackage[linesnumbered,ruled]{algorithm2e}

\usepackage{mathtools}
\newcommand{\bea}{\begin{eqnarray}}
\newcommand{\eea}{\end{eqnarray}}
\newcommand{\bean}{\begin{eqnarray*}}
\newcommand{\eean}{\end{eqnarray*}}

\newcommand{\sbinom}[2]{\left( \begin{array}{c} #1 \\ #2 \end{array} \right) }

\newcommand{\field}[1]{\mathbb{#1}}

\newcommand{\F}{\field{F}}

\newcommand{\cC}{{\cal C}}

\newcommand{\cO}{{\cal O}}

\newcommand{\cS}{{\cal S}}

\newcommand{\sG}{\script{G}}

\newcommand{\sP}{\script{P}}

\newcommand{\bfc}{{\boldsymbol c}}

\newcommand{\bfm}{{\boldsymbol m}}

\newcommand{\bfs}{{\boldsymbol s}}

\newcommand{\bfu}{{\boldsymbol u}}
\newcommand{\bfv}{{\boldsymbol v}}

\newcommand{\bfx}{{\boldsymbol x}}
\newcommand{\bfy}{{\boldsymbol y}}
\newcommand{\bfz}{{\boldsymbol z}}

\newcommand{\bfX}{{\mathbf X}}

\DeclareMathAlphabet{\mathbfsl}{OT1}{cmr}{bx}{it}
\newcommand{\uuu}{\kern-1pt\mathbfsl{u}\kern-0.5pt}
\newcommand{\vvv}{\kern-1pt\mathbfsl{v}\kern-0.5pt}

\newcommand{\myboxplus}{\kern1pt\mbox{\small$\boxplus$}}

\makeatletter \DeclareRobustCommand{\sbinom}{\genfrac[]\z@{}}
\makeatother
\newcommand{\G}[2]{\sbinom{{#1}\kern-1pt}{{#2}\kern-1pt}}
\newcommand{\Gq}[2]{\sbinom{{#1}\kern-0.25pt}{{#2}\kern-0.25pt}}
\newcommand{\Fq}{\smash{{\mathbb F}_{\!q}}}

\newcommand{\Ps}{\smash{{\sP\kern-2.0pt}_q\kern-0.5pt(n)}}
\newcommand{\sPs}{\smash{{\sP\kern-1.5pt}_q(n)}}
\newcommand{\Ptwo}{\smash{{\sP\kern-2.0pt}_2\kern-0.5pt(n)}}
\newcommand{\Ptwom}{\smash{{\sP\kern-2.0pt}_2\kern-0.5pt(m)}}
\newcommand{\Ptwonm}{\smash{{\sP\kern-2.0pt}_2\kern-0.5pt(n+m)}}
\newcommand{\Ptwoa}{\smash{{\sP\kern-2.0pt}_2\kern-0.5pt(1)}}
\newcommand{\Ptwob}{\smash{{\sP\kern-2.0pt}_2\kern-0.5pt(2)}}
\newcommand{\Ptwoc}{\smash{{\sP\kern-2.0pt}_2\kern-0.5pt(3)}}
\newcommand{\Ptwod}{\smash{{\sP\kern-2.0pt}_2\kern-0.5pt(4)}}
\newcommand{\Ptwoe}{\smash{{\sP\kern-2.0pt}_2\kern-0.5pt(5)}}
\newcommand{\Ptwof}{\smash{{\sP\kern-2.0pt}_2\kern-0.5pt(6)}}
\newcommand{\Ptwokm}{\smash{{\sP\kern-2.0pt}_2\kern-0.5pt(2k-1)}}
\newcommand{\Pone}{\smash{{\sP\kern-2.5pt}_2\kern-0.5pt(n{-}1)}}

\newcommand{\Gr}{\smash{{\sG\kern-1.5pt}_q\kern-0.5pt(n,k)}}
\newcommand{\Gi}{\smash{{\sG\kern-1.5pt}_q\kern-0.5pt(n,i)}}
\newcommand{\Gj}{\smash{{\sG\kern-1.5pt}_q\kern-0.5pt(n,j)}}
\newcommand{\Grmk}{\smash{{\sG\kern-1.5pt}_q\kern-0.5pt(n,n-k)}}
\newcommand{\Grdk}{\smash{{\sG\kern-1.5pt}_q\kern-0.5pt(2k,k)}}
\newcommand{\Grekappa}{\smash{{\sG\kern-1.5pt}_q\kern-0.5pt(n,e+1-\kappa)}}
\newcommand{\Grtwoekappa}{\smash{{\sG\kern-1.5pt}_q\kern-0.5pt(n,2e+1-\kappa)}}
\newcommand{\Gremkappa}{\smash{{\sG\kern-1.5pt}_q\kern-0.5pt(n,e-\kappa)}}
\newcommand{\Gn}{\smash{{\sG\kern-1.5pt}_2\kern-0.5pt(n,n{-}1)}}
\newcommand{\Gnq}{\smash{{\sG\kern-1.5pt}_q\kern-0.5pt(n,n{-}1)}}
\newcommand{\Gone}{\smash{{\sG\kern-1.5pt}_2\kern-0.5pt(n,1)}}
\newcommand{\Gqone}{\smash{{\sG\kern-1.5pt}_q\kern-0.5pt(n,1)}}
\newcommand{\GTwo}{\smash{{\sG\kern-1.5pt}_2\kern-0.5pt(n,k)}}
\newcommand{\GTwonk}[2]{{\smash{{\sG\kern-1.5pt}_2\kern-0.5pt({#1},{#2})}}}
\newcommand{\Gnk}{\smash{{\sG\kern-1.5pt}_2\kern-0.5pt(n,n{-}k)}}
\newcommand{\Greone}{\smash{{\sG\kern-1.5pt}_q\kern-0.5pt(n,e{+}1)}}
\newcommand{\Gretwo}{\smash{{\sG\kern-1.5pt}_q\kern-0.5pt(n,e{+}2)}}

\newcommand{\be}[1]{\begin{equation}\label{#1}}
\newcommand{\ee}{\end{equation}}

\newcommand{\Cref}[1]{Co\-rol\-la\-ry\,\ref{#1}}

\newcommand{\LCS}{\mathsf{LCS}}
\newcommand{\FLL}{\mathsf{FLL}}

\newtheorem{theorem}{Theorem}
\newtheorem{lemma}{Lemma}
\newtheorem{remark}{Remark}
\newtheorem{corollary}{Corollary}

\newtheorem{definition}{Definition}
\newtheorem{proposition}{Proposition}
\newtheorem{example}{Example}
\newtheorem{construction}{Construction}

\begin{document}

\author{%
  \IEEEauthorblockN{\textbf{Shubhransh~Singhvi}}

  \IEEEauthorblockA{%
                      International Institute of Information Technology, Hyderabad, India\\
                      \texttt{shubhranshsinghvi2001@gmail.com}}
}

\title{Optimally Decoding Two-Dimensional Reed-Solomon Codes Against Deletion Errors}
\date{}
\maketitle
\thispagestyle{empty}	
\pagestyle{empty}
%%%%%%%%

\begin{abstract}
Constructing Reed-Solomon (RS) codes that can correct insertion and deletion (ins-del) errors has been the focus of several recent studies. However, efficient decoding algorithms for such codes have received less attention and remain a significant open problem. In this work, we take a first step toward addressing this problem by designing a decoding algorithm for the case of $2$-dimensional RS codes that can correct deletions up to the half-Singleton bound and is optimal in terms of field operations.
\end{abstract}
%%%%%%%%

\section{Introduction}
In the last decade, channels that introduce insertion and deletion errors have attracted significant attention due to their relevance to DNA storage systems~\cite{anavy2019data, grass2015robust, organick2018random, pan2022rewritable, tabatabaei2020dna, yazdi2017portable}, where deletions and insertions are among the most dominant error types~\cite{heckel2019characterization, sabary2021solqc}. An insertion error occurs when a new symbol is added between two symbols of the transmitted word. Conversely, a deletion error occurs when a symbol is removed from the transmitted word. These errors affect the length of the received word. For example, over the binary alphabet, if $100110$ is transmitted, the received word might be $11011000$, which can result from three insertions (a $1$ at the beginning and two $0$s at the end) and one deletion (one of the $0$s at the beginning of the transmitted word). The study of communication channels with insertion and deletion errors is also relevant to various other applications, such as the synchronization of files and data streams~\cite{dolecek2007using}, and scenarios involving over-sampling or under-sampling at the receiver side~\cite{sala2016synchronizing}. VT codes, introduced by Varshamov and Tenengolts~\cite{varshamov1965code}, represent the first family of codes capable of correcting a single deletion or a single insertion~\cite{levenshtein1966binary}. Subsequent research extended these schemes to handle multiple deletion errors, as well as substitution and edit errors; see, e.g.,~\cite{brakensiek2017efficient, gabrys2018codes, gabrys2022beyond, guruswami2021explicit, schoeny2017codes, cheng2018deterministic, sima2020optimal, sima2021coding, smagloy2023single, sun2023sequence}. In some applications, such as DNA storage systems, the problem of list decoding has also been studied. For example, works such as~\cite{abu2021list, guruswami2020optimally, hayashi2020list, hanna2019list, liu2019list, wachter2017list} examine scenarios where the decoder receives a channel output and returns a (short) list of possible codewords, including the transmitted codeword. 

Even though significant progress has been made in recent years toward understanding the ins-del error model—both in identifying its limitations and constructing efficient codes—our understanding of this model still lags far behind that of codes designed to correct erasures and substitution errors. For further details, we refer the reader to the following excellent surveys: \cite{mitzenmacher2009survey, mercier2010survey, cheraghchi2020overview, haeupler2021synchronization}.

Linear codes provide a compact representation, efficient encoding algorithms, and, in many cases, efficient decoding algorithms. The rate $R$ of a code $\cC \subseteq \Fq^n$ is defined as $\frac{\log_q\left(\vert\cC\vert\right)}{n}$, which, for a linear code, is $\frac{k}{n}$, where $k$ is the dimension of the code. While linear codes are predominantly used for correcting Hamming errors, most constructions for handling ins-del errors are non-linear. This disparity can be attributed to findings in works such as \cite{abdel2007linear, cheng2023efficient}, which demonstrate that the maximal rate of a linear code capable of correcting ins-del errors is significantly lower than that of a non-linear code correcting the same number of such errors.

Consequently, exploring the performance of linear codes against the ins-del error model remains an important open problem and has been the focus of recent studies~\cite{cheng2023efficient, chen2022coordinate, con2022explicit, con2023reed, ji2023strict, cheng2023linear, liu2024optimal, con2024optimal, con2024random}. The works \cite{cheng2023efficient,chen2022coordinate, ji2023strict} have shown that the best one can hope for while designing linear codes capable of decoding ins-del errors is to achieve the \emph{half-Singleton bound}, given next. 

% \begin{theorem}{Half-Singleton bound} \cite[Corollary 5.2]{cheng2023efficient}. Every linear ins-del code which is
% capable of correcting a $\delta$ fraction of deletions has rate at most $\frac{1 - \delta}{2} +\smallO(1)$.
% \end{theorem}

% Theorem 3 (Half-Singleton bound; see [26, Corollary 5.2]). An [n, k]q linear code C can
% correct at most n − 2k + 1 insdel errors. Equivalently, LCS(C) ≥ 2k − 2.

\begin{theorem}[Half-Singleton bound; see {\cite[Corollary 5.2]{cheng2023efficient}}]
An $[n, k]_q$ linear code $\mathcal{C}$ can correct at most $n - 2k + 1$ ins-del errors.
\end{theorem}

Reed-Solomon (RS) codes \cite{reed1960polynomial} are among the most widely used families of codes in both theory and practice, with numerous applications, including QR codes, secret sharing schemes, space transmission, data storage, and more. Their ubiquity is largely due to their simplicity and the availability of efficient encoding and decoding algorithms. Consequently, understanding how RS codes perform against ins-del errors is a significant problem, one that has garnered considerable attention in recent years~\cite{safavi2002traitor, wang2004deletion, tonien2007construction, do2021explicit, liu20212, chen2022improved, con2023reed, con2024optimal, liu2024optimal, con2024random}. We next formally define RS Codes.
\begin{definition}[Reed-Solomon codes]
Let $\alpha_1, \alpha_2, \ldots, \alpha_n \in\F_q$ be distinct points in a finite field $\mathbb{F}_q$ of order  $q\geq n$. For $k\leq n$ the $[n,k]_q$ RS-code
defined by the evaluation vector $\bm{\alpha} = ( \alpha_1, \ldots, \alpha_n )$ is the set of codewords 
\[
\left \lbrace c_f = \left( f(\alpha_1), \ldots, f(\alpha_n) \right) \mid f\in \mathbb{F}_q[x],\deg f < k \right \rbrace \;.
\]
\end{definition}
Namely, a codeword of an $[n,k]_q$ RS-code is the evaluation vector of some polynomial of degree less than $k$ at $n$ predetermined distinct points. 
It is well known that the rate of $[n,k]_q$ RS-code is $k/n$ and the minimal distance, w.r.t. the Hamming metric, is $n-k+1$, the maximum possible by the \textit{Singleton bound} \cite{singleton1964maximum}. As such, RS codes are maximum distance separable(MDS) w.r.t. the Hamming metric.  

Notably, Con et al. in \cite{con2022explicit} have demonstrated that certain RS codes can achieve the half-Singleton bound.  However, the problem of efficiently handling the ins-del errors remains largely open. While RS codes, like other linear codes, benefit from efficient encoding algorithms, efficient decoding algorithms are not a direct consequence of linearity. Instead, they require a careful analysis of the code's algebraic structure.

We measure the time complexity in terms of field operations over $\mathbb{F}_q$, and we omit $\text{poly}(\log q)$ factors in our complexity notations.

\paragraph*{\textbf{Scope of this work}} Informally, the goal of this work is to design an efficient decoding algorithm for RS codes against ins-del errors, achieving a decoding radius as close as possible to the half-Singleton bound(ideally achieving the bound itself). For $k=2$, by carefully analyzing the structure of the $2$-dimensional RS code constructed by Roni et al.~\cite{con2024optimal}, we design a \textit{linear-time} decoding algorithm capable of correcting up to $n-3$ deletions, the maximum number of deletion errors that can be corrected. Specifically, 

\setcounter{theorem}{0}
\refstepcounter{theorem}
\begin{theorem}[Informal]\label{thm:main}
There exists a deterministic construction of an $[n,2]_q$-RS code over a field of size $q\sim \Theta(n^3)$, for which there is a linear-time decoding algorithm that can recover the transmitted codeword from any $n-3$ received symbols.
\end{theorem}

\paragraph*{\textbf{Outline}} In Section~\ref{sec:rel_work}, we review related work. Section~\ref{sec:notation_prelim} introduces the notation and presents preliminary results. In Section~\ref{sec:dec_algo}, we present the main contribution of this work—a linear-time decoding algorithm for the $2$-dimensional Reed-Solomon code constructed by Roni et al.~\cite{con2024optimal}. Section~\ref{sec:conclusion} then concludes the paper with a discussion on future work and open problems.

\section{Related Works} 
\label{sec:rel_work}
\subsection{Decoding RS Codes against Substitution Errors}
Designing efficient algorithms for decoding RS codes against substitution(Hamming) errors has been of immense importance and has received significant attention.  

An $[n,k]_q$ Reed-Solomon (RS) code is uniquely decodable up to the optimal radius $\rho = \frac{1 - R}{2}$, where $R \triangleq \frac{k}{n}$ is the rate of the code. This means that for every received word $\bfy \in \Fq^n$, there exists at most one codeword within Hamming distance $\rho n$. Over the years, several efficient decoding algorithms have been developed that achieve this unique decoding radius \cite{peterson1960encoding, gorenstein1961class, berlekamp1968algebraic, massey1969shift, sugiyama1975method, welch1983error, gao2003new, justesen2006complexity, bellini2011structure, fedorenko2002finding, lin2007fast, gao2010additive, wu2012reduced, lin2016fft, fedorenko2019efficient, tang2022new}.

% Berlekamp-Massey~\cite{berlekamp1968algebraic, massey1969shift} and Euclidean algorithms~\cite{sugiyama1975method} generate syndromes from the received codeword, which are then used to generate key-equations. Then, by solving the key-equations, the algorithms locate the erroneous positions and error magnitudes. For an $[n, k]$-RS code, the computational complexity of syndrome-based decoding is $\cO(n(n - k) + (n - k)^2)$. Welch and Berlekamp in \cite{welch1983error} presented a new key-equation and solving algorithm for decoding RS codes without computing the syndromes. However, the computational complexity of Welch-Berlekamp algorithm is also $\cO(n(n - k) + (n - k)^2)$. Over the years, faster approaches based on fast Fourier transforms or fast polynomial arithmetic techniques \cite{gao2003new, justesen2006complexity, bellini2011structure, fedorenko2002finding, lin2007fast, gao2010additive, wu2012reduced, lin2016fft, fedorenko2019efficient} have led to a unique-decoding algorithm by Tang et al. \cite{tang2022new} with computational complexity  of $\cO(n\log(n-k) + (n - k)\log^2(n-k))$.  

A code $\cC \subseteq \Fq^n$ is $(\rho, L)$-list decodable if, for every received word $\bfy \in \Fq^n$, there are at most $L$ codewords of $\cC$ within the Hamming distance $\rho n$ of $\bfy$. 
Due to Johnson-bound \cite{johnson1962new}, Reed–Solomon codes are $(\rho, L)$-list-decodable for error parameter $\rho = 1 - \sqrt{R} - \epsilon$ and list size $L = \cO\left(\frac{1}{\epsilon}\right)$. Sudan\cite{sudan1997decoding} developed an efficient algorithm to list-decode RS codes up to the radius $\rho = 1 - \sqrt{2R}$, which was later improved by Sudan and Guruswami \cite{guruswami1999improved} to efficiently
list-decode RS codes up to the Johnson radius $1 - \sqrt{R}$. Koetter and Vardy \cite{koetter2003algebraic} further modified the Guruswami-Sudan algorithm to incorporate soft-decoding.

For $\bfv\in \mathbb{F}_q^n$, let $B_t(\bfv)$ be the radius-$t$ ball centered at $\bfv$. A code $\cC$ is $(\rho n, N)$-sequence reconstructable \cite{levenshtein2001efficientA, levenshtein2001efficientB} if, for every codeword $\bfc \in \cC$, it can be uniquely reconstructed from the $N$ distinct channel outputs of $\bfc$, denoted by $Y \triangleq \left\{\bfy_{1}, \ldots, \bfy_{N}\right\} \subseteq B_{\rho n}(\bfc)$. To uniquely reconstruct $\bfc$, it must hold that, for all the other codewords $\bfc' \in \cC /\{\bfc\}$, $\bfc' \not\in \bigcap_{i=1}^{N}B_{\rho n}(\bfy_i)$. Recently, Singhvi et al. \cite{singhvi2024optimal} explored RS codes in the context of sequence reconstruction and proposed an efficient algorithm capable of decoding beyond the Johnson-bound.

\subsection{RS codes against ins-del errors}
The performance of RS codes against ins-del errors was first considered in \cite{safavi2002traitor} in the context of traitor tracing. In \cite{wang2004deletion}, the authors constructed a $[5, 2]$-RS code capable of correcting a
single deletion. Subsequently, in \cite{tonien2007construction}, an $[n, k]$-generalized RS code capable of correcting $\log_{k+1}(n) - 1$ deletions was constructed. In \cite{do2021explicit, liu20212}, the authors provided constructions of $2$-dimensional RS codes capable of correcting $n - 3$ ins-del errors. A breakthrough was achieved in \cite{con2023reed}, where the authors constructed the first linear codes that achieve the half-Singleton bound. Specifically, they demonstrated that certain $[n, k]$ RS codes achieve this bound with a field size of $n^{O(k)}$. In \cite{con2024optimal}, for RS codes with dimension $k=2$, the authors showed that the minimal field size required is $q=\Omega(n^3)$ and gave explicit code constructions matching the minimal field size. Furthermore, in a more recent study~\cite{con2024random}, the authors demonstrated that, with high probability, random Reed-Solomon codes approach the half-Singleton bound over a linear-sized alphabet.

\section{Notations and Preliminaries}
\label{sec:notation_prelim}
Let $[k] \triangleq \{1,2,\ldots,k\}$, where $k \in \mathbb{Z}_{+}$. Let $\Sigma_q$ be an alphabet of size $q$. For a set $\cS$ and $k \in \mathbb{Z}_{+}$, we denote the set of all possible subsets of $\cS$ of size $k$ by $\binom{\cS}{k}$. For a prime power $q$, we denote with $\Fq$ the field of size $q$. Let $\Fq^* \triangleq \Fq/\{0\}$. Let $a,b \in \Fq^{*}$, then by the abuse of notation, we use $\frac{a}{b}$ to refer to $a\cdot b^{-1}$. Let $\cC\subseteq \Fq^n$ be a linear code of dimension $k$, then it can be described as the image of a linear map, which, abusing notation, we also denote with $\cC$, i.e., $\cC : \Fq^k \rightarrow \Fq^n$. When $\cC\subseteq \F_q^n$ has dimension $k$ we say that it is an $[n,k]_q$ code (or an $[n,k]$ code defined over $\Fq$). The minimal distance of a code $\cC$ with respect to a metric $d(\cdot,\cdot)$ is defined as $\min_{\bfv, \bfu \in \cC,\bfu \neq \bfv}{d(\bfv,\bfu)}$.

In this work, we focus on codes against insertions and deletions. 

\begin{definition}
Let $\bfs$ be a string over the alphabet $\Sigma_q$. The operation in which we remove a symbol from $\bfs$ is called a \emph{deletion} and the operation in which we place a new symbol from $\Sigma_q$ between two consecutive symbols in $\bfs$, in the beginning, or at the end of $\bfs$, is called an \emph{insertion}.
\end{definition}

\begin{example} Let $\bfs =  \texttt{CAT} \in \Sigma_5^3$ (treating $\{C,A,T\}\subseteq\Sigma_5$). If we delete the second symbol, we obtain $\texttt{CT}$. If we insert a symbol $\texttt{G}$ (also in $\Sigma_5$) at the beginning, we obtain $\texttt{GCAT}$.
\end{example}

\begin{definition}
    A string $\bfy$ over the alphabet $\Sigma_q$ is a  \emph{subsequence} of $\bfx\in\Sigma_q^n$ if $\bfy$ can be obtained by deleting symbols from $\bfx$. Similarly, a string $\bfz$ over the alphabet $\Sigma_q$ is a  \emph{supersequence} of $\bfx\in\Sigma_q^n$ if $\bfx$ is a subsequence of $\bfz$. 
\end{definition}	

\begin{example}Let $\bfx = \texttt{BANANA}$. Then
$\texttt{BANA}$ is a subsequences of $\bfx$ and can be obtained by deleting the 2nd and 5th symbols of \texttt{BANANA}. Conversely, $\texttt{BANANAX}$ (where $\texttt{X}\in\Sigma_q$ is any symbol) is a supersequence of \texttt{BANANA} since \texttt{BANANA} appears as a subsequence in the longer string \texttt{BANANAX}.
\end{example}

% \begin{definition}
% The \textit{$k$-insertion ball} centred at ${\bfx\in\Sigma^n}$, denoted by $I_k(\bfx)\subseteq \Sigma^{n+k}$, is the set of all $k$-supersequences of $\bfx$. Similarly, the {\textit{$k$-deletion ball}} centred at ${\bfx\in\Sigma^n}$, denoted by $D_k(\bfx)\subseteq \Sigma^{n-k}$, is the set of all $k$-subsequences of~$\bfx$. 
% \end{definition}	

Let $\bfx, \bfx'$ be $n$-length strings over the alphabet $\Sigma_q$. The relevant metric for codes against insertions and deletions is the fixed-length Levenshtein(FLL) distance that we define next.

\begin{definition}
The \emph{FLL distance} between $\bfx$ and $\bfx'$, denoted by $\FLL(\bfx,\bfx')$, is the smallest integer $t$ such that $\bfx$ can be transformed to $\bfx'$ by $t$ insertions and $t$ deletions.  
\end{definition}

\begin{example} Let $\bfx=\texttt{TORN}$ and $\bfx'=\texttt{TRIM}$. We can go from \texttt{TORN} to \texttt{TRIM} by deleting the letters $\texttt{O,N}$ (two deletion) and inserting the letters $\texttt{I,M}$. Thus $\FLL(\texttt{TORN},\texttt{TRIM}) = 2$ since two deletions and two insertions suffice. 
\end{example}

One of the most fundamental parameters in any metric is the size of a ball with a given radius $t$ centered at a word $\bfx$. In the case of the FLL metric, the size of the ball depends on the word $\bfx$ \cite{sala2013counting, bar2022size, wang2024size, he2023size}, contrary to the Hamming metric. 

\begin{definition}
We denote the the longest common subsequence between any two sequences $\bfx,\bfx'\in\Sigma_q^*$ by $\LCS(\bfx,\bfx')$. 
\end{definition}

\begin{example} Let $\bfx = \texttt{ABCBDAB}$ and $\bfx'=\texttt{BDCABA}$. Then one longest common subsequence is $\texttt{BCBA}$ (of length 4), so $|\LCS(\bfx,\bfx')|=4$. 
\end{example}

It is well known that the ins-del correction capability of a code is determined by the maximal length of longest common subsequence of all pairs of distinct codewords. Specifically, 
\begin{lemma} \label{lem:ins-del-comb}
A code $\cC\subseteq \Fq^n$ can correct any combination of $t$ insertions and deletions if and only if, for any two distinct codewords $\bfc,\bfc'\in \cC$, it holds that $\left|\LCS(c,c')\right| < n - t$. Furthermore, it holds that $\FLL(\bfc,\bfc') = n -  \left|\LCS(\bfc,\bfc')\right| $.
\end{lemma}

We use $d_{F}(\cC)$ to denote the minimal distance of a code $\cC\subseteq \Fq^n$ w.r.t the FLL distance. The following corollary is a direct consequence of Lemma~\ref{lem:ins-del-comb}. 

\begin{corollary}
A code $\cC \subseteq \mathbb{F}_q^n$ can correct any pattern of at most $t$ deletions if and only if $d_F(\cC) = t+1$.
\end{corollary}

\begin{remark}
This work focuses exclusively on deletion errors. It is worth noting, however, that a decoding algorithm capable of correcting any $t$ deletions does not, in general, extend to correcting arbitrary combinations of $t$ insertions and deletions. Hence, the case of correcting arbitrary combinations of $t$ insertions and deletions is still an open problem.  
\end{remark}

Throughout this paper, we shall move freely between representations of vectors as strings and vice versa. Namely, we shall view each vector $\bfv=(v_1, \ldots, v_n)\in \Fq^n$ also as a string by concatenating all the symbols of the vector into one string, i.e., $(v_1, \ldots, v_n) \leftrightarrow v_1 \circ v_2 \circ \ldots \circ v_n$, where $\circ$ denotes concatenation. Thus, if we say that $\bfs$ is a subsequence of some vector $\bfv$, we mean that we view $\bfv$ as a string and $\bfs$ is a subsequence of that string. 

\medskip
\noindent\textbf{Example.} Let $\bfv=(2,0,3,2)\in \F_5^4$. We may view $\bfv$ as the string $\texttt{2032}$. A subsequence of this string is $\texttt{22}$ (obtained by deleting positions 2 and 3), which corresponds to the vector $(2,2)\in \F_5^2$. Conversely, if we start with a string $\texttt{2032}$ and delete symbol $\texttt{0}$, we end up with $\texttt{232}$, corresponding to the vector $(2,3,2)\in \F_5^3$.

\subsection{An Algebraic Condition}
In this section, we recall the algebraic condition presented in \cite{con2023reed}, which was used to characterize the ins-del error correcting capability of RS codes. 

We first make the following definitions:
We say that a vector of indices $I\in [n]^s$ is an \emph{increasing} vector if its coordinates are monotonically increasing, i.e., for any  $1\leq i<j\leq s$, $I_i<I_j$, where $I_i$ is the $i$th coordinate of $I$.  
%For a codeword $c$ of length $n$ and an increasing vector $I$, let $c_I$ be the restriction of $c$ to the coordinates with indices in $I$, i.e., $c_I=(c_{I_1},\ldots,c_{I_s})$. 
For any two distinct increasing vectors $I,J\in [n]^{2k-1}$, we define a (variant of) Vandermonde matrix of order $(2k-1)\times (2k-1)$ denoted by $V_{I,J}(\bfX)$ in the formal variables $\bfX =(X_1,\ldots,X_n)$ as follows:

\begin{equation} \label{eq:mat-lcs-eq}
\begin{pmatrix} 
1 & X_{I_1} & \ldots & X_{I_1}^{k-1}  & X_{J_1} &\ldots & X_{J_1}^{k-1} \\ 
1 & X_{I_2} & \ldots & X_{I_2}^{k-1}  & X_{J_2} &\ldots & X_{J_2}^{k-1} \\
\vdots &\vdots & \ldots &\vdots &\vdots &\ldots &\vdots \\
1 & X_{I_{2k-1}} & \ldots & X_{I_{2k-1}}^{k-1}  & X_{J_{2k-1}} &\ldots & X_{J_{2k-1}}^{k-1}\\
\end{pmatrix} .
\end{equation}
\vspace{0.2em}

\begin{proposition}
\cite[Proposition 2.1]{con2023reed} \label{prop:cond-for-RS} Let $\cC$ be an $[n,k]_q$ RS-code defined by an evaluation vector   $\bm{\alpha}=(\alpha_1,\ldots,\alpha_n)$.  If for every two increasing vectors $I,J\in [n]^{2k-1}$ that agree on at most $k-1$ coordinates, it holds that 
\[
\det(V_{I,J}(\bm{\alpha})) \neq 0\;,
\]
then $d_F(\cC) = n-2k+2$, i.e., $\cC$ achieves the half-Singleton bound.

% then the code can correct any $n-2k+1$ ins-del errors.
% Moreover, if the code can correct  any $n-2k+1$ ins-del errors, then the only possible vectors in $\text{Kernel}\left(V_{I,J}(\bm{\alpha})\right)$ are of the form $(0,f_1,\ldots,f_{k-1},-f_1,\ldots,-f_{k-1})$.
\end{proposition} 

For the ease of exposition, we say that an $[n,k]_q$ RS-code defined by an evaluation vector $\bm{\alpha}=(\alpha_1,\ldots,\alpha_n)$ satisfies the \emph{algebraic condition} if the evaluation vector $\bm{\alpha}$ satisfies the condition stated in Proposition \ref{prop:cond-for-RS}. 

\begin{definition}
Let $\cC$ be an $[n,2]_q$-RS code with evaluation vector $\bm{\alpha} = (\alpha_1, \ldots, \alpha_n) \in \mathbb{F}_q^n$ satisfying the algebraic condition.  
Define the map  
\[
\Gamma : \binom{\bm{\alpha}}{3} \to \mathbb{F}_q,\quad \Gamma(\alpha_i, \alpha_j, \alpha_k) \triangleq \frac{\alpha_i - \alpha_j}{\alpha_j - \alpha_k}
\]
for all $1 \le i < j < k \le n$.
\end{definition}

\begin{lemma}\label{lem:into_map}
Let $\cC$ be an $[n,2]_q$-RS code with evaluation vector $\bm{\alpha} = (\alpha_1, \ldots, \alpha_n) \in \mathbb{F}_q^n$ satisfying the algebraic condition. Then, the map $\Gamma : \binom{\bm{\alpha}}{3} \to \mathbb{F}_q$ is injective. 
\end{lemma}
\begin{proof}
Let $I, J \in [n]^3$ be any two increasing vectors. We next consider the following two cases based upon the number of shared coordinates.\\
\textbf{Case I:} $I, J \in [n]^3$ agree on at most one coordinate.\\
Since $\mathcal{C}$ satisfies the algebraic condition, from Proposition \ref{prop:cond-for-RS} the following holds:
\[
\left|
\begin{pmatrix}
1 & \alpha_{I_1} & \alpha_{J_1} \\ 
1 & \alpha_{I_2} & \alpha_{J_2} \\
1 & \alpha_{I_3} & \alpha_{J_3}
\end{pmatrix} 
\right|
\neq 0
\;,
\]
or equivalently, 
%This implies that 
\[
\frac{\alpha_{I_1} - \alpha_{I_2}}{\alpha_{I_2} - \alpha_{I_3}} \neq \frac{\alpha_{J_1} - \alpha_{J_2}}{\alpha_{J_2} - \alpha_{J_3}} \;. 
\]
\textbf{Case II:} $I, J \in [n]^3$ agree on exactly two coordinates.\\
Since the reaming coordinates are distinct, it is easy to verify that
\[
\frac{\alpha_{I_1} - \alpha_{I_2}}{\alpha_{I_2} - \alpha_{I_3}} \neq \frac{\alpha_{J_1} - \alpha_{J_2}}{\alpha_{J_2} - \alpha_{J_3}} \;. 
\]
Therefore, the map $\Gamma$ is injective.
\end{proof}

\section{Our Decoding Algorithm}
\label{sec:dec_algo}
In \cite{con2024optimal}, the authors designed explicit $[n,2]_q$-RS Codes, over a field size of $q =\cO(n^3)$, which can correct up to $n-3$ deletions. We describe one of the constructions presented in \cite{con2024optimal} below. 

\begin{construction}\cite[Proposition 2.3]{con2024optimal} \label{const:RS2-odd}
Let $\mathbb{F}_{q}$ be a finite field of characteristic $p>2$ and let $\mathcal{A}\subseteq \Fq^*$ be a subset of size $n$. Let $\gamma$ be a root of a degree $3$ irreducible polynomial over $\Fq$, and let the vector 
$\bm{\alpha} = (\alpha_1, \alpha_2, \ldots, \alpha_n)$
be some ordering  of the $n$ elements $\delta + \delta^2 \cdot \gamma, \delta \in \mathcal{A}$. Then, the $\left[n,2\right]$ RS-code defined over $\mathbb{F}_{q^3}$ with the evaluation vector $\bm{\alpha}$ can correct any $n-3$ ins-del errors. Furthermore, the blocklength can be as large as $q - 1$.
\end{construction}

% We consider the worst-case scenario in which $n-3$ symbols are deleted from the transmitted codeword $\mathbf{c}$. Let $\kappa_1, \kappa_2, \kappa_3$ denote the variables corresponding to the indices of the received symbols. Consequently, the received vector is given by $(c_{\kappa_1}, c_{\kappa_2}, c_{\kappa_3})$. 

We first present a decoding algorithm that can recover from any $n-3$ deletions in $\cO(n^3)$ time for any $[n,2]_q$-RS code that satisfies the algebraic condition. Then specifically for Construction \ref{const:RS2-odd}, we use the algebraic structure of the evaluation points to design an improved decoding algorithm that can recover from any $n-3$ deletions in \emph{linear}-time.

\subsection{Cubic-Time Decoding Algorithm}
Let $\bfm \in \Fq^2$ and $\bfc \in \cC$ be the codeword corresponding to the message vector $\bfm$. Let $\bfc$ be the transmitted codeword. Recall that each codeword symbol $c_i$ is obtained by evaluating a degree-one polynomial $f(x)=m_1 + m_2\,x$ at a known point $\alpha_i$.  In other words, 
\[
c_i = f(\alpha_i) = m_1 + m_2\,\alpha_i,
\]
where $(m_1,m_2)$ is the original message pair (both unknown to the decoder).  If exactly $n-3$ symbols are deleted, then the decoder receives a vector of length three:
\[
\left(c_{\kappa_1},\, c_{\kappa_2},\, c_{\kappa_3}\right),
\]
for some unknown indices $\kappa_1 < \kappa_2 < \kappa_3$.  Although the decoder does not know $m_1$, $m_2$, or the indices themselves, it can compute the ratio
\[
\beta \;=\; \frac{\,c_{\kappa_1} - c_{\kappa_2}\,}{\,c_{\kappa_2} - c_{\kappa_3}\,}. 
\]
Since each $c_{\kappa_j} = m_1 + m_2\,\alpha_{\kappa_j}$, the subtraction cancels $m_1$ and leaves
\[
c_{\kappa_1} - c_{\kappa_2} = m_2\,\bigl(\alpha_{\kappa_1} - \alpha_{\kappa_2}\bigr), 
\quad
c_{\kappa_2} - c_{\kappa_3} = m_2\,\bigl(\alpha_{\kappa_2} - \alpha_{\kappa_3}\bigr).
\]
Thus the ratio $\beta$ simplifies to
\[
\beta \;=\; \frac{\alpha_{\kappa_1} - \alpha_{\kappa_2}}{\alpha_{\kappa_2} - \alpha_{\kappa_3}},
\]
independent of $m_1$ and $m_2$.  In other words, the received vector can be used to identify exactly which three evaluation points $(\alpha_{\kappa_1},\alpha_{\kappa_2},\alpha_{\kappa_3})$ must have been used, provided that no two distinct triples of $\bm{\alpha}$ give the same `ratio'. 

From Lemma~\ref{lem:into_map}, we know that any RS code that satisfies the algebraic condition guarantees precisely this injectivity property. Hence, a decoder can simply compute $\beta$ from the received vector, and then search through all $\binom{n}{3}$ triples of indices $(J_1,J_2,J_3)$ to find the unique triple for which 
\[
\frac{\alpha_{J_1} - \alpha_{J_2}}{\alpha_{J_2} - \alpha_{J_3}} \;=\;\beta.
\]
Once those three indices are identified, the decoder knows exactly which points $(\alpha_{\kappa_j},c_{\kappa_j})$ lie on the original degree-one polynomial $f(x)$.  Recovering the message $(m_1,m_2)$ then amounts to fitting a degree‐one polynomial through any two of these received points. Finally, knowing $(m_1,m_2)$ allows reconstruction of all $n$ symbols by evaluating $f(\alpha_i)$ for $i=1,\dots,n$.

We now formally present the decoding algorithm and then verify it's correctness and time-complexity in the subsequent theorem. 

\begin{algorithm}[H]
\caption{Cubic-Time Decoding Algorithm}
\label{alg:dec_algA_formal}
\SetAlgoLined
\DontPrintSemicolon

\SetKwInOut{Input}{Input}
\SetKwInOut{Output}{Output}

\Input{Received vector $(c_{\kappa_1}, c_{\kappa_2}, c_{\kappa_3})$ (in order of appearance)\\
Evaluation vector $\bm{\alpha}=(\alpha_1,\dots,\alpha_n)\in\Fq^n$ satisfying the algebraic condition}
\Output{Recovered codeword $\bfc=(c_1,\dots,c_n)\in\cC$}

\medskip
Compute 
\[
\beta \;=\; \frac{\,c_{\kappa_1} - c_{\kappa_2}\,}{\,c_{\kappa_2} - c_{\kappa_3}\,} \;\in\;\Fq.
\]

\For{$J = (J_1,J_2,J_3)\in [n]^3$ with $1 \le J_1 < J_2 < J_3 \le n$}{
    Compute 
    \[
    \eta \;=\; \frac{\alpha_{J_1} - \alpha_{J_2}}{\alpha_{J_2} - \alpha_{J_3}} \;\in\;\Fq.
    \]
    \If{$\eta = \beta$}{
        Set $(\kappa_1,\kappa_2,\kappa_3) \;=\; (J_1,J_2,J_3)$\;
        \textbf{Break}\;
    }
}

Interpolate the unique degree‐$\le1$ polynomial $f(x) = m_1 + m_2\,x$ satisfying
\[
f(\alpha_{\kappa_j}) = c_{\kappa_j}, \quad j=1,2,3
\]
(using any two points).

\For{$i = 1,\dots,n$}{
    Set 
    \[
    c_i = f(\alpha_i).
    \]
}
\Return $\bfc = (c_1,\dots,c_n)$.
\end{algorithm}

% \noindent\textbf{Complexity in Plain Terms.}  Since there are $\binom{n}{3}=\cO(n^3)$ possible triples of indices, testing each triple’s ratio against the received $\beta$ takes $\cO(n^3)$ time overall.  After locating the correct triple, determining the line $f(x)$ and re‐generating the entire length‐$n$ codeword is only linear time in $n$.  Thus the decoding cost is dominated by the cubic‐time search over all triples.

% \begin{algorithm}
% \caption{Decoding Algorithm A}
% \label{alg:dec_algA}
% \SetAlgoLined
% \DontPrintSemicolon

% \SetKwInOut{Input}{input}
% \SetKwInOut{Output}{output}

% \Input{Received Vector; $(c_{\kappa_1}, c_{\kappa_2}, c_{\kappa_3})$}
% \Output{Transmitted codeword; $\bfc$}
% Set $\beta = \frac{c_{\kappa_1} - c_{\kappa_2}}{c_{\kappa_2} - c_{\kappa_3}}$ \;
% \For{$J \in [n]^3$}{
% Set $\eta = \frac{\alpha_{J_1} - \alpha_{J_2}}{\alpha_{J_2} - \alpha_{J_3}}$\;
% \If{$\eta = \beta$}
% {
% $\bfK_i = J_i$, where $i\in [3]$\;
% \textbf{Break}\;
% }
% }
% Interpolate $(c_{J_1}, c_{J_2}, c_{J_3})$ to find $\bfc$\;
% Return $\bfc$\; 
% \end{algorithm}

\begin{theorem} Let $\cC$ be an $[n,2]_q$-RS code with evaluation vector $\bm{\alpha} = (\alpha_1, \ldots, \alpha_n)$ satisfying the algebraic condition. Then, Algorithm~\ref{alg:dec_algA_formal} can recover the transmitted codeword in $\cO(n^3)$ time.  
\end{theorem}
\begin{proof}
Let $\bfm \in \Fq^2$ and $\bfc \in \cC$ be the codeword corresponding to the message vector $\bfm$. Therefore, we have that
\begin{align*}
c_i = m_1 + \alpha_i m_2,
\end{align*}
where $i\in[n]$. Furthermore, it follows that 
\begin{align*}
\Gamma(\alpha_{\kappa_1}, \alpha_{\kappa_2}, \alpha_{\kappa_3}) &= \frac{\alpha_{\kappa_1} - \alpha_{\kappa_2}}{\alpha_{\kappa_2} - \alpha_{\kappa_3}}  \\
&= \frac{\left((c_{\kappa_1}-m_1)m_2^{-1}\right) - \left((c_{\kappa_2}-m_1)m_2^{-1}\right)}{\left((c_{\kappa_2}-m_1)m_2^{-1}\right) -\left((c_{\kappa_3}-m_1)m_2^{-1}\right)}\\ 
&= \frac{c_{\kappa_1} - c_{\kappa_2}}{c_{\kappa_2} - c_{\kappa_3}} \triangleq \beta. 
\end{align*}
Therefore, we get that 
\begin{align}\label{eq:triplet_cond}
\frac{\alpha_{\kappa_1} - \alpha_{\kappa_2}}{\alpha_{\kappa_2} - \alpha_{\kappa_3}}  = \beta\;.
\end{align}
As shown in Lemma~\ref{lem:into_map}, the map $\Gamma$ is injective. Therefore, there is a unique solution to Equation~\eqref{eq:triplet_cond}. Since there are $\binom{n}{3}=\cO(n^3)$ possible triples of indices, testing each evaluation point triple's value under the map $\Gamma$ against  $\beta$ takes $\cO(n^3)$ time overall.  After locating the correct triple, determining the degree-one polynomial $f(x)$ and reconstructing the entire length‐$n$ codeword is only linear time in $n$.  Thus the decoding cost is dominated by the cubic‐time search over all triples.
\end{proof}

% \begin{algorithm}
% \caption{Decoding Algorithm B}
% \label{alg:dec_alg_linear}
% \SetAlgoLined
% \DontPrintSemicolon

% \SetKwInOut{Input}{input}
% \SetKwInOut{Output}{output}

% \Input{Received Vector; $(c_{\kappa_1}, c_{\kappa_2}, c_{\kappa_3})$}
% \Output{Transmitted codeword; $\bfc$}
% Set $\beta = \frac{c_{\kappa_1} - c_{\kappa_2}}{c_{\kappa_2} - c_{\kappa_3}}$ \;
% Let $\beta = a\gamma^2 + b\gamma + c$ \;
% Let $\beta\gamma = r\gamma^2 + s\gamma + t$ \;
% Let $\theta = \frac{a}{r}$\;
% Compute $\delta_{\kappa_2} = \dfrac{(b - \theta \left(c^2 + s - 2 c t \theta + t^2 (\theta)^2\right))}{(2 (c + c^2 - 2 c t \theta + t \theta (-1 + t \theta)))} $\;
% Compute $\delta_{\kappa_1} = \delta_{\kappa_2}\left(1 + 2c-2t\theta\right) + \theta\left(c-t\theta\right)$\;
% Compute $\delta_{\kappa_3} = - \delta_{\kappa_2} - \theta$ \; 
% Set $\alpha_{\bfK_i} = \delta_{\bfK_i}  + \delta_{\bfK_i}^2\gamma$, where $i\in[3]$\;
% \For{$i \in[3]$}
% {
% \For{$j \in [n]$}
% {
% \If{$\alpha_j = \alpha_{\bfK_{i}}$}
% {$\bfK_i = j \triangleq J_i$}
% }
% }
% Interpolate $(c_{J_1}, c_{J_2}, c_{J_3})$ to find $\bfc$\;
% Return $\bfc$\; 
% \end{algorithm}

\subsection{Linear-time Decoding Algorithm}

Recall that in Construction \ref{const:RS2-odd}, each evaluation point is of the form 
\[
\alpha_i \;=\; \delta_i \;+\; \delta_i^2\,\gamma,
\]
where $\delta_i\in\Fq^*$ and $\gamma\in\Fq$ is a root of a minimal polynomial of degree $3$ over $\Fq$. As before, we form the ratio 
\[
\beta 
\;=\; 
\frac{\,c_{\kappa_1} - c_{\kappa_2}\,}{\,c_{\kappa_2} - c_{\kappa_3}\,}
\;=\;
\frac{\alpha_{\kappa_1} - \alpha_{\kappa_2}}{\alpha_{\kappa_2} - \alpha_{\kappa_3}}
\]
because each $c_{\kappa_j}=m_1+m_2\,\alpha_{\kappa_j}$.  Writing 
\[
\beta \;=\; a\,\gamma^2 + b\,\gamma + c,
\quad
\beta\,\gamma \;=\; r\,\gamma^2 + s\,\gamma + t,
\quad a,b,c,r,s,t\in\Fq,
\]
and substituting $\alpha_{\kappa_j}=\delta_{\kappa_j}+\delta_{\kappa_j}^2\gamma$ into 
\[
\frac{\alpha_{\kappa_1} - \alpha_{\kappa_2}}{\alpha_{\kappa_2} - \alpha_{\kappa_3}} \;=\;\beta
\]
yields a degree‐$<3$ polynomial in $\gamma$.  Equating each coefficient of $1,\gamma,\gamma^2$ to zero gives a system of quadratic equations in $\delta_{\kappa_1},\delta_{\kappa_2},\delta_{\kappa_3}\in\Fq$.  Solving these yields explicit formulas:
\[
\theta \;=\;\frac{a}{r}, 
\quad
\delta_{\kappa_2} 
\;=\; 
\frac{\,b - \theta\bigl(c^2 + s - 2c\,t\,\theta + t^2\,\theta^2\bigr)\,}
     {\,2\bigl(c + c^2 - 2c\,t\,\theta + t\,\theta(-1 + t\,\theta)\bigr)\,},
\]
\[
\delta_{\kappa_3} = -\,\delta_{\kappa_2} - \theta,
\quad
\delta_{\kappa_1} 
= 
\delta_{\kappa_2}\bigl(1 + 2c - 2\,t\,\theta\bigr) + \theta\bigl(c - t\,\theta\bigr).
\]
Since each $\delta_{\kappa_j}\in\Fq^*$, a prebuilt hash‐table of $\delta_i\mapsto i$ identifies $\kappa_j$ in $\cO(1)$ time.  Thus, instead of testing $\binom n3$ triples, we solve for $(\delta_{\kappa_1},\delta_{\kappa_2},\delta_{\kappa_3})$ in $\cO(1)$ field operations and then interpolate the unique degree-one polynomial through the three points in $\cO(n)$ time.

We now formally present the decoding algorithm and then verify it's correctness and time-complexity in the subsequent theorem.

\begin{algorithm}[H]
\caption{Linear-Time Decoding Algorithm}
\label{alg:dec_alg_linear}
\SetAlgoLined
\DontPrintSemicolon
\SetKwInOut{Input}{Input}
\SetKwInOut{Output}{Output}

\Input{Received Vector $(c_{\kappa_1},\,c_{\kappa_2},\,c_{\kappa_3})$ and $\gamma \in \Fq$\\
Hash‐table mapping each $\delta_i\in\Fq^*$ to index $i$, and points $\alpha_i=\delta_i+\delta_i^2\gamma$.}
\Output{Recovered codeword $\bfc=(c_1,\dots,c_n)$}

Compute 
\[
\beta = \frac{\,c_{\kappa_1} - c_{\kappa_2}\,}{\,c_{\kappa_2} - c_{\kappa_3}\,}\;.
\]
Express
\[
\beta = a\,\gamma^2 + b\,\gamma + c,\quad
\beta\,\gamma = r\,\gamma^2 + s\,\gamma + t,
\]
with $a,b,c,r,s,t\in\Fq$.\;

Set $\theta = a/r$.\;

Compute 
\[
\delta_{\kappa_2} 
= 
\frac{\,b - \theta\bigl(c^2 + s - 2c\,t\,\theta + t^2\,\theta^2\bigr)\,}
     {\,2\bigl(c + c^2 - 2c\,t\,\theta + t\,\theta(-1 + t\,\theta)\bigr)\,},\quad
\delta_{\kappa_3} = -\,\delta_{\kappa_2} - \theta,\quad
\delta_{\kappa_1} = \delta_{\kappa_2}\bigl(1 + 2c - 2\,t\,\theta\bigr) + \theta\bigl(c - t\,\theta\bigr).
\]

Lookup each $\kappa_j$ from $\delta_{\kappa_j}$ in the hash‐table (each in $\cO(1)$).\;

Interpolate the unique degree one-polynomial $f(x)=m_1 + m_2\,x$ through 
\[
(\alpha_{\kappa_j},\,c_{\kappa_j}),\quad j=1,2,3,
\]
and set $c_i = f(\alpha_i)$ for $i=1,\dots,n$.\;

\Return $\bfc$.
\end{algorithm}

\setcounter{theorem}{0}
\refstepcounter{theorem}
\begin{theorem}[Formal]
Let $\cC$ be an $[n,2]_{q^3}$-RS code defined over $\mathbb{F}_{q^3}$ with evaluation vector $\bm{\alpha} = (\alpha_1, \ldots, \alpha_n)$ as specified in Construction \ref{const:RS2-odd}. Then, Algorithm~\ref{alg:dec_alg_linear} can recover the transmitted codeword from any $n-3$ received symbols in $\cO(n)$ time. 
\end{theorem}

\begin{proof}
Let $\bfm \in \Fq^2$ and $\bfc \in \cC$ be the codeword corresponding to the message vector $\bfm$. Therefore, we have that
\begin{align*}
c_i = m_1 + \alpha_i m_2,
\end{align*}
where $i\in[n]$. Furthermore, it follows that 
\begin{align*}
\Gamma(\alpha_{\kappa_1}, \alpha_{\kappa_2}, \alpha_{\kappa_3}) &= \frac{\alpha_{\kappa_1} - \alpha_{\kappa_2}}{\alpha_{\kappa_2} - \alpha_{\kappa_3}}  \\
&= \frac{\left((c_{\kappa_1}-m_1)m_2^{-1}\right) - \left((c_{\kappa_2}-m_1)m_2^{-1}\right)}{\left((c_{\kappa_2}-m_1)m_2^{-1}\right) -\left((c_{\kappa_3}-m_1)m_2^{-1}\right)}\\ 
&= \frac{c_{\kappa_1} - c_{\kappa_2}}{c_{\kappa_2} - c_{\kappa_3}} \triangleq \beta\;.
\end{align*}
Therefore, we get that 
\begin{align}\label{eq:triplet_cond_B}
\frac{\alpha_{\kappa_1} - \alpha_{\kappa_2}}{\alpha_{\kappa_2} - \alpha_{\kappa_3}}  = \beta\;.
\end{align}
As shown in Lemma~\ref{lem:into_map}, the map is injective. Therefore, there is a unique solution to Equation~\eqref{eq:triplet_cond_B}.

Next, from Construction \ref{const:RS2-odd}, for $i\in [n]$, we have $\alpha_i = \delta_i + \delta_i^2 \gamma$, where recall that $\delta_i \in \Fq^{*}$. Upon making this substitution in \eqref{eq:triplet_cond_B}, we get 
\begin{small}
\begin{align*}
\delta_{\kappa_1} - \delta_{\kappa_2} + (\delta_{\kappa_1}^2- \delta_{\kappa_2}^2) \gamma 
+ \beta \big( \delta_{\kappa_3} - \delta_{\kappa_2} + (\delta_{\kappa_3}^2- \delta_{\kappa_2}^2) \gamma \big)  = 0 \;.
\end{align*}
\end{small}

Write $\beta = a\gamma^2 + b\gamma + c$ and $\beta\gamma = r\gamma^2 + s\gamma + t$, where $a,b,c,r,s,t \in \F_q$. Observe that the LHS is a polynomial in $\gamma$ of degree less than $3$ over $\Fq$.
Namely, 
$
p_0(\bm{\delta}) + p_1(\bm{ \delta})\cdot \gamma +p_2(\bm{ \delta})\cdot \gamma^2 = 0
$,
where,
\begin{align*}
p_0(\bm{\delta}) &= \delta_{\kappa_1} - \delta_{\kappa_2} + c(\delta_{\kappa_3} - \delta_{\kappa_2} ) + t(\delta_{\kappa_3}^2 - \delta_{\kappa_2}^2)\;,\\
p_1(\bm{\delta}) &= \delta_{\kappa_1}^2 - \delta_{\kappa_2}^2 + b(\delta_{\kappa_3} - \delta_{\kappa_2} ) + s(\delta_{\kappa_3}^2 - \delta_{\kappa_2}^2)\;,\\
p_2(\bm{\delta}) &= a(\delta_{\kappa_3} - \delta_{\kappa_2} ) + r(\delta_{\kappa_3}^2 - \delta_{\kappa_2}^2)\;.
\end{align*}
Next, by the definition of $\gamma$,  $p_i(\bm{\delta}) = 0$ for  $i=0,1,2$. We first solve the system of equations for the case of $a = 0$ to serve as a sanity check for the general case.  For $a=0$, we get 
\begin{align*}
\delta_{\kappa_2} &= -\delta_{\kappa_3}, \\
\delta_{\kappa_1} &= - (2c+1)\delta_{\kappa_3}, \\ \delta_{\kappa_3} &= \frac{-b}{2c^2 + 2c}~\text{or}~0. 
\end{align*}
Since $\delta_{\kappa_1}, \delta_{\kappa_2}, \delta_{\kappa_3} \in \Fq^{*}$, we must have that 
$$
\delta_{\kappa_3} = \frac{-b}{2c^2 + 2c}.
$$
Let $\theta = \frac{a}{r}$. For the general case, we get  
\begin{align*}
\delta_{\kappa_3} &= - \delta_{\kappa_2} - \theta,\\ 
\delta_{\kappa_1} &= \delta_{\kappa_2}\left(1 + 2c-2t\theta\right) + \theta\left(c-t\theta\right),\\
\delta_{\kappa_2} &= \dfrac{(b - \theta \left(c^2 + s - 2 c t \theta + t^2 (\theta)^2\right))}{(2 (c + c^2 - 2 c t \theta + t \theta (-1 + t \theta)))}~\text{or}~\frac{-a}{2r}.
\end{align*}
Since $\delta_{\kappa_1} \neq \delta_{\kappa_2} \neq \delta_{\kappa_3}$, we must have that
$$
\delta_{\kappa_2} = \dfrac{(b - \theta \left(c^2 + s - 2 c t \theta + t^2 (\theta)^2\right))}{(2 (c + c^2 - 2 c t \theta + t \theta (-1 + t \theta)))}. 
$$
Hence, the decoding algorithm can find the evaluation points in $\cO(1)$ time. Once the evaluation points are known, the interpolation step to find the transmitted codeword can be performed in $\cO(n)$ time as the dimension of the code is $\cO(1)$.  
\end{proof}

Hence, the RS code construction by Roni et al. \cite{con2024optimal} is not only optimal in terms of field size, but it can also be decoded with optimal time complexity.

\section{Conclusion}\label{sec:conclusion}
In this work, we focused on the case of $2$‐dimensional RS code and by leveraging the algebraic structure of the evaluation points in the code construction of Roni et al.~\cite{con2024optimal}, we designed  a \emph{linear‐time} decoding algorithm that corrects up to $n-3$ deletions, maximum possible due to the half‐Singleton bound for ins‐del errors.  

Looking ahead, a natural direction is to explore whether similar techniques can be extended to higher dimensions ($k > 2$). Two key challenges arise in this setting:
\begin{enumerate}
    \item \textbf{Field‐Size Optimality.} For general $k$, the half‐Singleton bound implies that any $[n,k]_q$-RS code can correct at most $n - 2k + 1$ ins-del errors. However, the minimum field size $q$ required to achieve this bound remains unknown.
    \item \textbf{Generalized Injective Map.} Our decoding strategy relies critically on an injective map defined over triples of evaluation points. Extending this approach to $k > 2$ would require a map from $(2k-1)$ received symbols to $\F_q$ that remains injective over the relevant domain. The existence and efficient construction of such a map are currently open questions.
\end{enumerate}

\section*{Acknowledgments}
The author would like to acknowledge Roni Con and Lalitha Vadlamani for fruitful discussions on the problem. 

\newpage
\bibliographystyle{plain} 
\bibliography{ref}

\end{document}